\documentclass{endm}
\usepackage{endmmacro}
\usepackage{graphicx}

\newlength{\parindentsave}

\begin{document}

\begin{verbatim}\end{verbatim}\vspace{2.5cm}

\begin{frontmatter}

\title{Efficient and Perfect domination on circular-arc graphs}

\author{Min Chih Lin$^{3,1}$\thanksref{thxoscar}\thanksref{oscaremail} \qquad\qquad Michel J Mizrahi}
\address{CONICET, Instituto de C\'alculo and Departamento de Computaci\'on\\ Universidad de Buenos Aires, Buenos Aires, Argentina}


\author{Jayme L Szwarcfiter\thanksref{thxjayme}\thanksref{jaymeemail}}
\address{Inst de Matem\'atica, COPPE and NCE  \\ Universidade Federal do Rio de Janeiro,
Rio de Janeiro, Brazil}

\thanks[oscaremail]{Email:
   \href{mailto:oscarlin@dc.uba.ar} {\texttt{\normalshape
   \{oscarlin,mmizrahi\}@dc.uba.ar}}}
\thanks[jaymeemail]{Email:
   \href{mailto:jayme@nce.ufrj.br} {\texttt{\normalshape
   jayme@nce.ufrj.br}}}
   
\thanks[thxoscar]{Partially supported by UBACyT Grant 20020120100058, PICT ANPCyT
Grants 2010-1970 and 2013-2205.}
\thanks[thxjayme]{Partially supported by CNPq, CAPES and FAPERJ, brazilian research agencies. Presently visiting the Instituto Nacional de Metrologia, Qualidade e Tecnologia", Brazil.}

%

\begin{abstract}
Given a graph $G = (V,E)$, a \emph{perfect dominating set} is a subset of vertices $V' \subseteq V(G)$ such that each vertex $v \in V(G)\setminus V'$ is dominated by exactly one vertex $v' \in V'$. An \emph{efficient dominating set} is a perfect  dominating set $V'$ where $V'$ is also an independent set.
These problems are usually posed in terms of edges instead of vertices. Both problems, either for the vertex or edge variant, remains NP-Hard, even when restricted to certain graphs families.
We study both variants of the problems for the circular-arc graphs, and show efficient algorithms for all of them.
\end{abstract}

\end{frontmatter}

\section{Introduction}\label{intro}
 
Given a graph $G = (V,E)$, a {\bf perfect dominating set} is a subset of vertices $V' \subseteq V(G)$ such that each vertex $v \in V(G)\setminus V'$ is dominated by exactly one vertex $v' \in V'$.
An {\bf efficient dominating set} is a perfect vertex dominating set $V'$ where $V'$ is also an independent set.
Every graph $G$ contains a perfect dominating set, for instance, take $V(G)$. But not every graph contains an efficient vertex dominating set.
These problems consists in searching the sets with minimum number of vertices. All of them are NP-hard, even when restricted to certain graph families.
The weighted version of these problems, where each vertex $v$ has a weight assigned $\omega(v)$, consists on finding a perfect 
vertex dominating set where the sum of the weights is minimum. We denote these problems as Minimum Weighted Perfect Vertex Domination (MWPVD), 
Minimum Weighted Efficient Vertex Domination (MWEVD). 
We denote the edge-versions of these problems as Minimum Weight Perfect Edge Domination (MWPED) and Minimum Weight Efficient Edge Domination (MWEED). Efficient edge dominating sets are also known as dominating induced matchings,
and denoted as DIM's.
Note that for these edge-versions the dominating set consists of edges instead of vertices, hence the weights are on the edges, and the adjacency of two edges
is defined as two edges that shares a vertex.
We say a \emph{pendant} vertex (also known as \emph{leaf}) is one whose degree is exactly one. 
In this paper we show results for the weighted perfect domination problem, and for the efficient domination problem, restricted to circular-arc graphs.

\section{Circular-Arc graphs}
The following definitions and results come from \cite{DuranLMS06}

Given a circular-arc model  $\mathcal{M}=(C, \mathcal{A})$ where $\mathcal{A}=\{A_1=(s_1,t_1),\dots, A_n=(s_n,t_n)\}$, two points $p,p' \in C$ are equivalent if $\mathcal{A}(p)=\mathcal{A}(p')$. 
The $2n$ extreme points from the $n$ arcs of $\mathcal{A}$ divide the circle $C$ in $2n$ segments of the following types: (i) $(s_i,t_j)$ (ii) $[t_i,t_j)$ (iii) $(s_i,s_j]$ (iv) $[t_i,s_j]$. We say the segments of type (i) are \textit{intersection segments}.
It is easy to see that all points inside one of the $2n$ segments are equivalent. 

\begin{corollary}\cite{DuranLMS06}
There are at most $2n$ distinct $\mathcal{A}(p)$.
\end{corollary}

\begin{lemma}\cite{DuranLMS06}\label{lemmaLMS}
Given a CA model $\mathcal{M}=(C,\mathcal{A})$, if there are no two or three arcs of $\mathcal{A}$ that covers the entire circle $C$ then $\mathcal{M}$ is an HCA model.
\end{lemma}


Given a circular-arc model $\mathcal{M}=(C,\mathcal{A})$, the following algorithms can be achieved in $O(n)$ time:

\begin{itemize}
\item Search a universal arc. A universal arc has common intersection with every arc from $\mathcal{A}$.
\item Search two arcs $A_i, A_j$ such that $A_i \cup A_j =C$.
\item Search three arcs $A_i, A_j, A_k$ such that $A_i \cup A_j \cup A_k=C$.
\item Search a point $p\in C$ such that $|\mathcal{A}(p)|$ is maximum or minimum.

\item If $max_{p \in C}|\mathcal{A}(p)|=2$ and $min_{p \in C}|\mathcal{A}(p)|=1$ and do not exists two arcs from $\mathcal{A}$ that covers the entire circle $C$ then exists
an induced cycle $C_k$ with $k \ge 3$ such that arcs corresponding to the vertices cover the circle $C$ in $\mathcal{M}$ and the rest of the arcs from $\mathcal{A}$
are pairwise disjoint and are contained in exactly one of the arcs from $C_k$. Thus each arc is either part of the $C_k$ or a leaf with a parent arc from $C_k$.
It is possible to identify each arc from the $C_k$.
\end{itemize}

Note that $max_{p \in C}|\mathcal{A}(p)|>min_{p \in C}|\mathcal{A}(p)|$ if $\mathcal{A}\ne \emptyset$. 
For instance, let $A = (s,t) \in \mathcal{A}$, $\mathcal{A}(s) \subset \mathcal{A}(s+\epsilon)$ and $A\in \mathcal{A}(s+\epsilon)\setminus\mathcal{A}(s)$. 
Hence $max_{p \in C}|\mathcal{A}(p)|\geq |\mathcal{A}(s+\epsilon)|>|\mathcal{A}(s)| \geq min_{p \in C}|\mathcal{A}(p)|$.


We consider the four variants of the mentioned problems for circular-arc graphs, we
propose efficient algorithms to solve them, except for the MWEVD for which already exists a linear time algorithm. 
We assume that our input is a circular-arc graph $G$ and a weight function $\omega$ over the vertices or edges depending on the problem we were solving.
For simplicity, we may use the circular-arc model from $G$. For these cases, there is an implicit previous step which is applying a linear time algorithm
\cite{McConnell} in order to obtain a circular-arc model from $G$.
For MWEVD and MWEED, we consider the function $\omega$ nonnegative. 

\section{Minimum weighted efficient vertex domination}

The minimum weighted efficient vertex domination problem (MWEVD)
on a graph $G$ can be expressed as an instance of the minimum weighted
dominating set problem (MWDS) on the same graph $G$ making some minor
adjustments on the way described in \cite{Br-Le-Ra} for unweighted version of MWEVD.
There is an $O(n+m)$ time algorithm to solve MWDS for circular-arc graphs \cite{Chang:98}. 
Hence, MWEVD can be solved for circular-arc graphs in linear time.

\section{Minimum weighted efficient edge domination}

The Minimum weighted efficient edge domination problem (MWEED) for a graph $G$ is equivalent to either:
\begin{enumerate}
	\item MWEVD for the line graph $L(G)$
	\item MWIS for the square of the line graph $L^2(G)$
\end{enumerate}

If $G$ is a circular-arc graph, then the graph $G' = L^2(G)$ is also a circular-arc graph. The graph $G'$ has exactly $m$ vertices and 
up to $O(m^2)$ edges. The algorithm from \cite{spinrad2003efficient} solves MWIS problem in linear time for circular-arc graphs. 
It is known that a graph that admits a DIM is $K_4$-free. This property imposes a bound on the number of edges in circular-arc graphs, thus the algorithm complexity
bound can be improved from $O(m^2)$ to $O(n^2)$.

\begin{lemma}\label{lem:bound1}
\cite{LinMS14} Every $K_4$-free graph $G$ where $n \geq 2$ such that $G$ is either chordal or interval graph has at most $2n-3$ edges. 
\end{lemma}

\begin{lemma}\label{lem:bound2}
Every $K_4$-free graph $G$ where $|V(G)| \geq 2$ such that $G$ is circular-arc graph has at most $2n$ edges. 
\end{lemma}

\begin{proof}
If $|V(G)|=1$ is trivially true, and if $G$ is an interval graph the property is satisfied by lemma \ref{lem:bound1}.
Thus, suppose $|V(G)| \geq 4$ and $G$ is not an interval graph.
If exists $p$ such that $|\mathcal{A}(p)| = 0$ then $G$ is an interval graph and if $|\mathcal{A}(p)|\geq 4$ then $G$ contains a $K_4$. 
In consequence, $1 \leq | \mathcal{A}(p)  | \leq 3$ for any point $p$. Let $p$ be the point with maximum value of $| \mathcal{A}(p) |$. Clearly, $p$ belongs to an intersection segment $(s_i,t_j)$ where $| \mathcal {A}(s_i)| = | \mathcal{A}(p) |-1$. Hence, $| \mathcal {A}(p)|\geq 2$ and we can restrict the study to $2 \leq |\mathcal{A}(p)| \leq 3$.
It is easy to see that at least one of these exists.  
We analyze each possible case:

\begin{itemize}
	\item $|\mathcal{A}(p)| = 2$: Let $(s_i,t_j)$ the intersection segment that contains $p$:
	
	\vspace{3mm}
	\begin{enumerate}
		\item $i = j$: In this case $\mathcal{A}(p) = \{A_i,A_k\}$ where $A_k$ contains arc $A_i$. If we cut arc $A_k$ at point $s_i$ converting
		it in two different arcs $(s_k,s_i - \epsilon)$ and $(s_i+\epsilon,t_k)$ we get an interval graph $G'$ with $n+1$ intervals. By Lemma ~\ref{lem:bound1}
		the number of edges of $G'$ is at most 2(n+1)-3 = 2n-1 edges. Since $G'$ has at least the same amount of edges that $G$ we conclude that $|E(G)| \leq 2n-1$.
		
		\item $i \ne j$: In this case $\mathcal{A}(p) = \{A_i,A_j\}$. If we replace $A_i$ by $(s_i,p-\epsilon)$ and $A_j$ by $(p + \epsilon, t_j)$, we get
		an interval graph $G'$ with $n$ vertices. However the edge that connects $A_i,A_j$ have been lost unless $A_i \cup A_j$ cover the entire circle.
		Therefore $|E(G')| \leq 2n-3$ and $|E(G)| \leq 2n-3+1 = 2n-2$.
	\end{enumerate}
	
	\vspace{3mm}
	\item $|\mathcal{A}(p)| = 3$: \\
	\begin{enumerate}
		\item $i=j:$ In this case $\mathcal{A}(p) = \{A_i,A_k,A_p\}$ where $A_k,A_p$ contain arc $A_i$. If we cut arcs $A_k,A_p$ at point $s_i$, converting
		each one into two different arcs, we get an interval graph $G'$ with $n+2$ intervals and at least one more edge since $A_k,A_p$ have
		been converted to two different edges. By Lemma ~\ref{lem:bound1} the number of edges of $G'$	 is at most $2(n + 2) - 3 = 2n + 1$, hence $|E(G)| \leq 2n$
		\item $i \ne j$: In this case $\mathcal{A}(p) = \{A_i,A_j,A_k\}$ where $A_k$ contains the segment $(s_i,t_j)$. 
		We can replace $A_i$ by ($s_i,p-\epsilon)$ and $A_j$ by $(p+\epsilon,t_j)$ (and may lose an edge unless $A_i \cup A_j$ covers the entire circle). 
		We cut the arc $A_k$ in either $s_i$ or $t_j$, any one works. We get an interval graph $G'$ with $n+1$ vertices and at least $|E(G)|-1$ edges. 
		By lemma ~\ref{lem:bound1} the number of edges of $G'$ is at most $2(n + 1) - 3 = 2n - 1$. Hence $|E(G)| \leq 2n$.
	\end{enumerate}
	
\end{itemize}

The graph $\overline{3K_2}$ has 6 vertices and 12 edges, hence it shows tightness of the bound.
Therefore any algorithm to solve DIM problem can first check the amount of edges to ensure the existence of a DIM.
Since $m \in O(n)$ then any $O(n+m)$ time algorithm for circular-arc graphs can be easily converted to an $O(n)$ time algorithm.
\end{proof}

A linear-time algorithm to solve MWEED for general graphs given a fixed dominating set was presented in \cite{DBLPMin}. If there exists a set of at most three arcs that covers the entire circle, then there is a dominating set of size at most $3$, thus the problem can be solved using the mentioned algorithm in linear time.
We assume a model $\mathcal{M}$ without a set $S$ that covers the entire circle such that $|S| \leq 3$.  Therefore $\mathcal{M}$ is a Helly circular-arc (HCA) model and the original graph $G$ is an $HCA$.

We analyze the model $\mathcal{M}$ (each case can be implemented in linear time):

\begin{description}
	\item \textbf{(a) $max_{p \in C}|\mathcal{A}(p)|\geq 4$:} Then $G$ is not $K_4$-free, hence it does not admit a DIM
	
	\item \textbf{(b) $min_{p \in C}|\mathcal{A}(p)|=0$:} Then $\mathcal{M}$ is an interval model.	Algorithm from \cite{LuKT02} can be applied
	
	\item \textbf{(c) $max_{p \in C}|\mathcal{A}(p)|=2$:} 
	

	

Since $max_{p \in C}|\mathcal{A}(p)| > min_{p \in C}|\mathcal{A}(p)|$, then $min_{p \in C}|\mathcal{A}(p)| = 1$.
We can locate an induced cycle $C_{k \geq 4}$ on $G$ such that arcs from those vertices covers the entire circle $C$
in the model $\mathcal{M}$, and the rest of arcs from $A$ are pairwise-disjoint and are contained in exactly one arc from $C_k$.
If a vertex $v \in C_k$ is connected with a set of pendant vertices $W$, then $v$ should be colored with black in order to
get a coloring of the graph that represents a valid DIM.
Note that the only edge that belongs to a minimum weighted DIM will be an edge of minimum weight
between $v$ and $W$. Therefore, all edges between $v$ and $w \in W$ such that its weight is not minimum can be removed from the
graph $G$ in order to obtain $G'$. Every minimum weighted DIM from $G$ is a valid minimum weighted DIM in $G'$. 
Note that if exists $k$ edges with minimum weight, then any $k-1$ of those edges can be erased.
In the new graph $G'$ every vertex $v \in C_k$ contains at most one pendant vertex $w$ in his neighborhood.
Since HCA is an hereditary property, $G'$ is still an HCA graph, and every clique from $G'$ is a $K_2$, therefore $K(G') = L(G')$.
The graph $K(G')$ is HCA if $G'$ is HCA and it contains $O(n)$ vertices and $O(n)$ edges since maximum 
$\Delta(L(G')) \leq 4$.
The $K(G')$ model can be obtained within $O(n)$ time \cite{LinMSS08}. 
Therefore MWEED can be solved for $G'$, by solving MWEVD for $L(G')$, which can be done in linear time.


\item \textbf{(d) $max_{p \in C}|\mathcal{A}(p)|=3$:} 


Denote $\mathcal{A}(p) = \{A_1, A_2, A_3\}$. Then $v_1,v_2,v_3$ form a triangle in $G$, which means that any DIM of $G$ must
contain exactly one of the three edges of this triangle: $\{v_1v_2,v_2v_3,v_3v_1\}$. By removing point $p$
from $C$ (every arc $(s,t)$ containing $p$ is now two disjoint arcs $(s,p-\epsilon)$ and $(p+\epsilon,t)$), we obtain an interval graph $G_p$, 
in which the vertices $v_1,v_2,v_3$ of $G$ correspond to two triangles formed by $u_1,u_2,u_3$ and $u_{n+1}, u_{n+2}, u_{n+3}$, respectively. So,
any DIM of $G_p$ must contain one edge of each of these triangles. Furthermore, our task would become
simpler if $u_1,u_2,u_3$ and $u_{n+1}, u_{n+2}, u_{n+3}$ would correspond to triangles having neither common nor adjacent
vertices. The latter condition is fulfilled because there are no two arcs that covers the entire circle.

\begin{lemma}\label{lt}
The triangles $\{u_1,u_2,u_3\}$ and $\{u_{n+1}, u_{n+2}, u_{n+3}\}$ have neither common nor adjacent vertices.
\end{lemma}
\begin{proof}
Clearly, $u_1,u_2,u_3, u_{n+1}, u_{n+2}, u_{n+3}$ are six different vertices. 
Hence there are no common vertex in both triangles. Suppose that $u_i$ is adjacent to $u_{n+j}$, $1\leq i,j\leq 3$. Consider the following cases.

\begin{description}
\item[(a) $i=j$,] which means $A_i$ covers the circle $C$, that is absurd.
\item[(b) $i\ne j$,] in this case, $A_i$ and $A_j$ cover the circle which is a contradiction.
\end{description}
\end{proof}

So, we now consider that the triangles $u_1,u_2,u_3$ and $u_{n+1},u_{n+2}, u_{n+3}$ contain neither common nor adjacent vertices.
We will apply algorithm from \cite{LuKT02} to $G_p$, which solves MWEED for chordal graphs (hence interval graphs).
The algorithm will be executed three times. In each of the applications, the graph $G_p$ remains the same, but the weighting
changes, to $\Omega_1$, $\Omega_2$ and $\Omega_3$, respectively.
In order to avoid edges of triangles $\{u_1,u_2,u_3\}$ and $\{u_{n+1},u_{n+2},u_{n+3}\}$ 
incident to $u_i$ or $u_{n+i}$ to be included in the DIM, in the weighting $\Omega_i$, $1 \leq i \leq 3$, 
we assign to each of these edges a high weight, for instance more than twice the sum of all weights of the edges of $G$. 
All the other edges of $G_p$ remain the same as in the corresponding edges of $G$. 
If the value of the minimum weighted DIM of $G_p$ is above that high weight assigned to the edges we want to avoid, then we know that
$G$ contains no DIM. Otherwise, we have solved our problem, and we only need to subtract the duplicated
weight among the edges of the triangles $u_1,u_2,u_3$ and $u_{n+1}, u_{n+2}, u_{n+3} $ which is part of the solution 
for $G_p$, but not for $G$.

The following are the formal definitions of the weights. Let
$A_1,A_2,A_3$ be the arcs containing $p$. Assign to $G_p$ the
weighting $\Omega_i$, $1 \leq i \leq 3$, which defines a weight
$\omega_i(u_ju_k)$ for each edge $u_ju_k \in E(G_p)$, as follows.
\bigskip
\\
For $1 \leq i,j \leq 3 \leq k \leq n$ and $u_ju_k \in E(G_p)$, \\
$\omega_i(u_ju_k) := \omega(v_jv_k)$
\\
\\
For $1 \leq i,j \leq 3 \leq k \leq n$ and $u_{n+j}u_k \in E(G_p)$, \\
$\omega_i(u_{n+j}u_k) := \omega(v_jv_k)$ 
\\
\\
For $1 \leq i,j,k \leq 3$, $i\ne j$, $j \ne k$ and $k \ne i$, \\
$\omega_i(u_ju_k) := \omega_i(u_{n+j}u_{n+k}) := \omega(v_jv_k)$
\\
\\
For $1 \leq i,j \leq 3$ and $i\ne j$, \\
$\omega_i(u_iu_j) := \omega_i(u_{n+i}u_{n+j}) := 1 + 2\sum \omega$
\\
\\

\begin{lemma} Let $max|\mathcal{A}(p)| = 3$, $\mathcal{A}(p) = \{A_1,A_2,A_3\}$. Then
$$dim_{\Omega}(G) = \left\{\begin{array}{cl}
\infty &
min_{1 \leq i \leq 3} \{dim_{\Omega_i}(G_p)\}> 2 \sum \omega \\
min_{1 \leq i \leq 3}\{dim_{\Omega_i}(G_p) - \omega_i(u_ju_k)\}^\dagger
 &
\mbox{\emph{otherwise}}
\end{array}\right.$$
$\dagger$ \emph{ where} $1 \leq j,k \leq 3; i\neq j \neq k \neq i$
\end{lemma}

\begin{proof}
Let $\mathcal{A}(p) = \{A_1,A_2,A_3\}$. Cut the circle at point
$p$ and consider the model ${\mathcal M}$ of the interval graph
$G_p$. For lemma \ref{lt}, $\{u_1,u_2,u_3\}$ and $\{u_{n+1},u_{n+2},u_{n+3}\}$ are triangles with
neither common nor adjacent vertices.

Suppose $G$ has a DIM $M$, weighted by
$\Omega$, having total weight $\omega' \leq 2 \sum \omega$. We
know that exactly one edge of the triangle $\{v_1,v_2,v_3\}$, say
$v_1v_2$, belong to $M$. Then choose the weighting $\Omega_3$ for
$G_p$ and consider the following subset of edges $M_p \subseteq
E(G_p)$.

$M_p = \{u_1u_2, u_{n+1}u_{n+2}\} \cup \{u_iu_j \in E(G_p) |
vf_iv_j \in M, 3 < i,j \leq n\} $
\\
We claim the $M_p$ is a DIM for $G_p$.
Since $M$ is a matching, $M_p$ is clearly so. Assume by contrary
that it is not induced, and let the violating edge be $u_iu_j$,
where $u_i,u_j$ are vertices incident to distinct edges of $M_p$.
Since $M$ is an induced matching of $G$, it follows that $i \in
\{1,2,3\}$ and $j \in \{n+1, n+2,n+3\}$. The latter means that
$u_1,u_2,u_3$ and $u_{n+1}, u_{n+2}, u_{n+3}$ have a common or an
adjacent vertex, a contradiction. Consequently, $M_p$ is indeed an
induced matching. It remains to show that it is dominating. Let
$S$ be the subset of vertices of $G_p$, not incident to the edges
of $M_p$. Let $u_i,u_j \in S$, $i \neq j$. Because $M$ is a
dominating matching, the only possibility for $u_i,u_j$ to be
adjacent is $i = 3$ and $j = n+3$, which contradicts $A_3$ not
covering the circle. Then $M_p$ is a DIM
having weight equal to $\omega' + \omega(v_1,v_2)$.

Conversely, assume that $G_p$ has a DIM
weighted by, say $\Omega_3$, having weight $ \leq 2 \sum \omega$.
We know that exactly one edge of each of the triangles
$\{u_1,u_2,u_3\}$ and $\{u_{n+1}, u_{n+2}, u_{n+3}\}$ belong to $M_p$.
Furthermore, the edges of these triangles which are incident
either to $u_3$ or $u_{n+3}$ all have weight $ > 2 \sum \omega$.
Therefore any dominating set of edges with weight $\leq 2 \sum
\omega$ contains the edges $u_1u_2$ and $u_{n+1}u_{n+2}$. Let

 $M = \{v_1v_2\} \cup \{v_iv_j \in E(G) | u_iu_j \in M_p, 3 < i,j \leq n\} $
\\
We claim that $M$ is a DIM of $G$.
Clearly, for any $v_i,v_j \in V(G)$, $1 \leq i \leq 3$ and $3 < j
\leq n$, $v_iv_j \in E(G)$ if and only if $u_iu_j \in E(G_p)$ or
$u_{n+i}u_j \in E(G_p)$. Consequently, $M_p$ being a DIM of $G_p$ implies that $M$ is a DIM of $G$. Furthermore, the weight of $M$ is precisely the
weight of $M_p$ less $\omega_3(u_1u_2)$, since it was counted
twice in $M_p$. The lemma follows. 
\end{proof}


\end{description}
\vspace{-3mm}
\section{Minimun weighted perfect vertex Domination}
An $O(n+m)$ time algorithm to solve MWPVD for interval graphs was presented in \cite{ChangL94a}. The same paper shows the only known algorithm to solve MWPVD for circular-arc graphs in $O(n^2 + nm)$ time. Note that any efficient vertex dominating set is also a perfect vertex dominating set.

Given a circular-arc graph $G$, a circular-arc model $\mathcal{M}$ from $G$ can be obtained in $O(n+m)$ time and universal arcs can be
identified in $O(n)$ time. If a universal arc exists, we can solve the problem using the following procedure:



\subsection{Minimum weighted perfect vertex Domination}~\label{approachMWPVD2}

\begin{lemma}\label{universalVertex}
Given a graph $G=(V,E)$ where $u_1 \in V(G)$ is a universal vertex:
	\begin{enumerate}
		\item If $G$ contains another universal vertex $u_2 \ne u_1$, then the unique perfect vertex dominating sets of $G$ are: $V$, $\{u_i\}$ where $u_i$ may be any
		universal vertex.
		
		\item If $u_1$ is the unique universal vertex from $G$, then every perfect vertex dominating set from $G$ contains $u_1$. Moreover, each solution can be
		computed on the following way: Let $\{G_1 = (V_1,E_1), \ldots, G_k= (V_k,E_k)$ the connected components from $G \setminus \{u_1\}$.
		Then every perfect vertex dominating set $D$ from $G$ should verify:
		
			\begin{enumerate}
				\item A vertex $w \in V_i$ belongs to $D$ if and only if $V_i(G) \in D$	
				\item A subset $V' \subseteq V$ that verifies (a) and contains $u_1$ is a perfect vertex dominating set from $G$
			\end{enumerate}

	\end{enumerate}
\end{lemma}

\begin{proof}
The proof is separated in two cases:
\begin{enumerate}
\item If $G$ contains universal vertex $u_2 \ne u_1$: \\
It is clear that $V$ and each universal vertex $\{u_i\}$ are perfect vertex dominating sets from $G$. 
Assume there exists another perfect vertex dominating set $D$. If $|D|=1$ then it should be a universal vertex and was one of the above mentioned sets.
Hence $D$ contains at least two vertices and there is at one least vertex $w \not\in D$. 
Since $|D|>1$ then every universal vertex $u_i$ should belong to $D$, otherwise the vertex $u_i$ will be dominated by more than one vertex inside $D$, 
which contradicts the definition of perfect vertex domination.
Moreover, every vertex $w \not\in D$ will be dominated by at least two universal vertices from $D$, then again, this contradicts the definition of the set $D$.
Therefore, there is not any other perfect vertex dominating set.

\item If $u_1$ is the unique universal vertex from $G$: \\
Assume exists a perfect vertex dominating set $D$ that does not contains $u_1$. It is clear that $D$ should have at least two vertices since
there is not another universal vertex, but then the vertex $u_1$ is dominated by $|D|>1$ vertices, which is absurd. Thus every perfect vertex dominating set
contains $u_1$.
Suppose $D$ is a perfect vertex dominating set such that $v \in D$, $w \not\in D$ and $v,w \in V_i$. 
Since $v,w$ belongs to the same connected component $G_i$, there exists a pair of adjacent vertices $v',w' \in V_i$ such that $v' \in D$ and $w' \not\in D$.
Thus $w'$ is dominated by $u_1$ and $v'$, which is a contradiction. 
Thus every perfect vertex dominating set satisfies (a).
Let $D$ be a subset of vertices that satisfies (a) and contains $u_1$ and $w \not\in D$.  
It is easy to see that $u_1$ dominates $w$. The set of vertices $N(w) \setminus \{u_1\}$ belongs to the same connected component. By (a), none of them
are in $D$. Therefore, $u_1$ is the unique vertex from $D$ that dominates $w$
\end{enumerate}
\end{proof}

If $G$ contains at least two universal vertices, then the amount of perfect vertex dominating sets is at most $O(n)$, and the one with minimum weight
can be obtained in linear time. 
In case $G$ contains exactly one universal vertex the minimum weighted perfect vertex dominating set can be obtained with the following procedure:

$D:= \{u_1\}$. We define the weight of a set of vertices as the sum of the weights of its vertices.
If the weight of a connected component $G_i$ of $G \setminus \{u_i\}$ is negative, then $D:= D \cup V_i$.


Thus we assume $G$ does not contain a universal vertex.

It is possible to solve MWEVD in linear time for circular-arc graphs. Hence we can save the best efficient vertex dominating set as a candidate solution, and
search for the perfect vertex dominating sets that are not efficient vertex dominating sets. We determine in $O(n)$ time the point $p$ such that $|\mathcal{A}(p)|$ is minimum, and according to this value a different approach can be used:



\begin{description}

\item \textbf{(i) $|\mathcal{A}(p)|=0$,} In this case $\mathcal{M}$ is an interval model and $G$ an interval graph. 
A linear-time algorithm \cite{ChangL94a} can be applied to solve MWPVD.

\vspace{4mm}
 
\item \textbf{(ii) $|\mathcal{A}(p)|\ge 2$,} thus for every point $p' \in C$, $|\mathcal{A}(p')| > 2$.
Let $D$ be a perfect vertex dominating set which is not an efficient vertex dominating set, and with minimum weight. 
There should be two adjacent vertices $v,w \in D$, otherwise $D$ is an efficient vertex dominating set.
Let $A_v,A_w$ the corresponding arcs. Note that $A_v \cap A_w \ne \emptyset$, and let $q \in A_v \cap A_w$. 
For any arc $A_z \in \mathcal{A}(q)$ the vertex $z$ (corresponding to $A_z$) should belong to $D$, otherwise (otherwise it will be dominated by two vertices). 
Let $q' \in C$ the first point from $q$ in clock-wise order such that $\mathcal{A}(q') \ne \mathcal{A}(q)$ (hence $q$ and $q'$ belong to different segments). 
We will show that for any arc $A_z \in \mathcal{A}(q')$, the vertex $z$ corresponding to $A_z$ should be in $D$
We consider two cases according to segment finalization $q$.

\begin{itemize}
	\item If it has open-end at $t_u$, then $\mathcal{A}(q') = \mathcal{A}(q) \setminus \{A_u\}$. Thus, every arc 
	$A_z \in \mathcal{A}(q') \subset \mathcal{A}(q)$, the vertex $z$ corresponding to $A_z \in D$.
	
	\item If it has close-end at $s_u$, then $\mathcal{A}(q)=\mathcal{A}(q')\setminus \{A_u\}$. 
	Thus $\mathcal{A}(q')$ contains arcs from $\mathcal{A}(q)$ and an additional arc $A_u$. Since $A_u$ is adjacent to arcs from $\mathcal{A}(q)$ and 
	$|\mathcal{A}(q)|>1$ then $A_u$ should be in $D$, since $\mathcal{A}(q)$ belongs to $D$ and $A_u$ could not have more than one adjacent from $D$, unless is part of $D$.
\end{itemize}

Let $q: = q'$ and apply iteratively the same procedure $2n$ times, this is the amount of different segments. This shows that all vertices are in $D$.
Thus $D = V$, and the solution will be the best among $\sum_{v \in V}\omega(v)$ and the best solution previously found.

\vspace{4mm}

\item \textbf{(iii) $|\mathcal{A}(p)|=1$,}  In this case $\mathcal{A}(p)=\{A_v\}$ where $A_v=(s_v,t_v)$ corresponds to a vertex $v\in V$. 
Recall that $\mathcal{A}$ has no universal arc. 
We generate three interval models ($\mathcal{M}_1,\mathcal{M}_2$ y $\mathcal{M}_3$) from $\mathcal{M}$ and replace $A_v$ with 
$A_{v^-}=(s_v,p-\epsilon)$ and $A_{v^+}=(p+\epsilon,t_v)$, adding arcs for each one of them and modifying the weight function as
we describe next:

\begin{description}
\item \textbf{($\mathcal{M}_1$)} We replace $A_v$ with $A_{v^-}$ and $A_{v^+}$. We also add $A_{w^-}=(p-1.5*\epsilon,p-0.5*\epsilon)$ and $A_{w^+}=(p+0.5*\epsilon,p+1.5*\epsilon)$, assigning weights to the vertices that corresponds to new arcs: $\omega(v^-):=\omega(v^+):=0.5*\omega(v)$, $\omega(w^-):=\omega(w^+):=\infty$. 

\item \textbf{($\mathcal{M}_2$)} We replace $A_v$ with $A_{v^-}$ y $A_{v^+}$. We also add $A_{w^-}=(p-1.5*\epsilon,p-0.5*\epsilon)$, assigning weights to the vertices that corresponds to the new arcs: $\omega(v^-):=\omega(v^+):=\infty$, $\omega(w^-):=0$. 

\item \textbf{($\mathcal{M}_3$)} We replace $A_v$ with $A_{v^-}$ y $A_{v^+}$. We also add $A_{w^+}=(p+0.5*\epsilon,p+1.5*\epsilon)$, assigning weights to the vertices that corresponds to the new arcs: $\omega(v^-):=\omega(v^+):=\infty$, $\omega(w^+):=0$.
\end{description}

It is easy to see that described models are interval models. Hence MWPVD can be solved for those three interval models using a linear time algorithm from \cite{ChangL94a}.

We show how to map the perfect vertex dominating set $D_i$ from $\mathcal{M}_i$, $1 \leq i \leq 3$, to a perfect vertex dominating set $D$ from $G$ maintaining
the weight.
\begin{itemize} 

\item The weight of $D_1$ is bounded, thus $w^-,w^+ \not\in D_1$. Since both are leaves, they should be dominated by their parents $v^-$ y $v^+$, respectively.
Then $v^-,v^+\in D_1$.
In this case, we take $D=(D_1\setminus \{v^-,v^+\})\cup\{v\}$. Note that the weight of $D$ and $D_1$ is the same since $\omega(v^-)+\omega(v^+)=\omega(v)$.

Each vertex from $v \in G$ such that $v \notin D$ is dominated by exactly one vertex from $D_1$. Each dominating vertex remains except for $v^-$ and $v^+$.
But vertices dominated by $v^-$ and $v^+$ are now dominated by $v$. Hence $D$ is a perfect vertex dominating set from $G$.

\item The weight of $D_2$ is bounded, hence $v^-,v^+ \not\in D_2$. Since $w^-$ is a leave $v^-\not\in D_2$, then $w^-\in D_2$.
The vertex $v^-$ is dominated by $w^-$, thus the rest of the neighbors of $v^-$ are not in $D_2$.

In this case, let $D=D_2\setminus \{w^-\}$. Note that $D$ and $D_2$ has the same weight since $\omega(w^-)=0$.

Each vertex $w\not\in D$ is dominated by exactly one vertex from $D$. All these vertices, except for $v$ were dominated by vertices from $D_2$ which remains
at $D$, except for $w^-$, but this vertex dominates only $v^-$ which is not a vertex from $G$.
$v^+$ is dominated by a vertex from $D=D_2\setminus \{w^-\}$. 
Clearly, this dominating vertex dominates $v$ and is the unique vertex that dominates $v$. As a consequence, $D$ is a perfect vertex dominating set from $G$.

\item The weight of $D_3$ is bounded, this case is symmetric to the previous case. Let $D=D_3\setminus \{w^+\}$ be a perfect vertex dominating set from $G$ with the same weight of $D_3$.
\end{itemize}
\end{description}

Now the prefect vertex dominating set $D$ from $G$ should be maped to a perfect vertex dominating set with the same weight of $D$ in some model $\mathcal{M}_i$, $1\leq i \leq 3$. 

The following describes the mapping:

\begin{itemize}
\item If $v \in D$, $D_1=(D\setminus \{v\}) \cup \{v^-,v^+\}$ in the model $\mathcal{M}_1$. Then $D_1$ and $D$ has the same weight.
$D_1$ is a perfect vertex dominating set in $\mathcal{M}_1$, unless there exists a vertex $z\not\in D_1$  such that is dominated by $v^-$ and $v^+$. 
In this case, the arcs  $A_v$ and $A_z$ (arc corresponding to $z$) covers the entire circle $C$. 
Clearly, $z\not\in D$ and is dominated by $v\in D$. But $D$ is a perfect vertex dominating set from $G$, thus $v$ is a universal vertex.
If this is not the case, then exists a vertex $w$ which is not adjacent to $v$ and the corresponding arc $A_w$ should be contained in $A_z \setminus A_v$,
but since $z$ is dominated by $v$, all the intersecting arcs with $z$ corresponds to vertices outside $D$. Then $w$ can not be dominated by any of the vertices
from $D$ because every intersecting arc of $A_w$ intersects $A_z$. Absurd. Therefore $v$ is a universal vertex, but it contradicts the hipothesis that no universal
vertices exists at $G$. Hence, it does not exists the vertex $z$ and $D_1$ is a perfect vertex dominating set from $\mathcal{M_1}$.

\item If $v \notin D$, $v$ is dominated by $z \in D$ and $t_v \in A_z$, where $A_z$ corresponds to vertex $z$.
It is clear that  $A_v\not\subset A_z$ since $p\in A_v\setminus A_z$. 
In this case, $D_2=D\cup \{w^-\}$. Again, $D$ and $D_2$ has the same weight. 
If $D_2$ is not a perfect vertex dominating set from $\mathcal{M}_2$, then $v^-$ is dominated by $w^- \in D_2$ and another vertex $u \in D_2$.
Moreover, $u \in D$ and dominates $v$ in $G$ but $v$ was dominated by $z \in D$. Hence $z=u$ and $A_z$ and $A_v$ covers the entire circle $C$.
Applying the same reasoning of the previous case we can conclude that $z$ is a universal vertex and contradicts the hypothesis that $G$ do not contain universal
vertices.
Therefore, $D_2$ is a perfect vertex dominating set in $\mathcal{M}_2$.

\item If $v\not\in D$, $v$ is dominated by $z\in D$ and $s_v\in A_z$. It's symmetric to the previous case, we obtain that $D_3=D\cup \{w^+\}$ is a perfect vertex dominating set in $\mathcal{M}_3$ with the same weight that $D$.
 
\end{itemize}



\section{Minimum weighted perfect edge domination}

We give an $O(n+m)$ time algorithm to solve MWPED for circular-arc graphs. 
To the best of our knowledge there is no known polynomial time algorithm  to solve this problem on circular-arc graphs. 
It is proved in \cite{LuKT02} the NP-completeness of unweighted version of this problem for bipartite graphs. 

\begin{theorem}\cite{LuKT02}
There is an $O(n+m)$ time algorithm to solve MWPED on chordal graphs
\end{theorem}

\begin{corollary}\cite{LuKT02}
There is an $O(n+m)$ time algorithm to solve MWPED on interval graphs
\end{corollary}

\begin{definition}
Given a graph $G=(V,E)$ and a perfect edge dominating set $E' \subseteq E$ from $G$, we denote $D=\{v\in V: vw \in E'\}$ the vertices incident to an edge from $E'$. We can define the following 3-coloring for the vertices of $G$: 
The black vertices $B=\{v\in D: N[v] \subseteq D\}$, the gray vertices $R=D\setminus B$ and the white vertices $W=V(G)\setminus D$. 
\end{definition}

The following properties can be easily checked:

\begin{description}
\item \textbf{(P1)} Each gray vertex has exactly one non-white neighbor while the rest of his neighborhood are white vertices. (the gray vertex has degree at least 2)

\item \textbf{(P2)} If $v \in W$, then $N(v) \subseteq R$. Hence $W$ is an independent set.
\end{description}

It is easy to see that for any 3-coloring of vertices that satisfies properties (P1) and (P2) $E'=\{vw\in E: vw \in B \cup R\}$ is a perfect edge dominating set
and for any 3-coloring of $G$ that satisfies (P1) and (P2), if it contains $K_p$, with $p\geq4$, then vertices from $K_p$ should be black. 

Any efficient edge dominating set of a graph $G$ (if it exists), is also a perfect edge dominating set from $G$. 

Given a circular-arc graph $G=(V,E)$, we show how to solve MWPED in linear time. 
First, solve MWEED in $O(n)$ time. If there is a solution (DIM), we save the one with minimum weight as a candidate solution. Therefore the candidates that should be explored are the perfect edge dominating sets that are not DIM.
Note that the set $E$ is also a candidate solution.

We obtain circular-arc model $\mathcal{M}=(C,\mathcal{A})$ from $G$. We solve MWPED for $G$ according to the following properties from $\mathcal{M}$.

\begin{description}
\item \textbf{(a)} There are 2 arcs, $A_v=(s,v,t_v),A_w=(s,w,t_w)\in \mathcal{A}$ such that $A_v\cup A_w = C$. Let $E'\ne E$ be a perfect edge dominating set which is not a DIM. 



It is clear that $E'$ determines a 3-coloring of the vertices from $V$ that verifies (P1) and (P2). Let $v$ and $w$ be the corresponding vertices to $A_v$ and $A_w$. 
The following possibilities form a valid color combination of $v$ and $w$ in order to satisfy (P1) and (P2):

\vspace{4mm}
\begin{description}

\item[(black,black).] It is clear that any other arc $A_z$ from the model that corresponds to the vertex $z \in V$ has common intersection with $A_v$ or $A_w$.

Using (P2), $z$ is not white. Hence, there is no white vertex in $V$. Therefore, $B=V$ and $E=E'$, absurd. This combination is not valid.

\vspace{4mm}

\item[(black,gray) and (gray,black).] 
Since both are symmetric, we consider one of them: (black,gray).
The vertex $w$ is gray and must verify (P1). The only non-white neighbor is $v$, hence the rest of his neighborhood should be white.
By (P2) the rest of neighbors of $v$ should not be white. Then there are no common vertices of $v$ and $w$.
Then for each vertex we know the color it must have, since every arc $A_z$ intersects with either $A_v$ or $A_w$.
If there is no conflict with the colors each vertex have and verifies (P1) and (P2), we can save this solution as one more candidate solution.
\vspace{4mm}

\item[(gray,gray).] Every vertex $z$ (different from $v$ and $w$) should be adjacent to $v$ and/or $w$.  
Since $v$ and $w$ are gray adjacent vertices and verify (P1), $z$ is a white vertex. Thus $E'$ contains the edge $vw$ only, thus $E'$ is a DIM, which 
contradicts the election of $E'$. Therefore this color combination is not valid.

\vspace{4mm}

\item[(gray,white) and (white,gray).] 

The vertex $v$ is gray and should verify (P1), hence it must have exactly one non-white neighbor $u$. 
On the other hand, vertices from $N(w)$ are gray because (P2). We analyze the following two cases: 
\vspace{4mm}

\begin{enumerate}
\item $A_u\subset A_v$: 
Every neighbor from $u$ is also neighbor from $v$, hence are all white. Vertices that are not neighbors from $v$ are neighbors from $w$ since their corresponding
arcs are contained at $A_w$ and should induce a matching since they satisfy (P2).
It is easy to see that $E'$ of this combination is exactly the edges from the induced matching plus the edge $vu$. Thus $E'$ is a DIM, which contradicts the 
election of $E'$. As a consequence, this combination is invalid for this case.
 
\item $A_u \cap A_w\ne\emptyset$, Then $u$ is a gray vertex (P2). By (P1) vertices from $N(u) \setminus \{v\}$ are white.
Vertices from $N(w)$ should be gray (P2). Thus there are no black vertices in $G$. Therefore $E'$ is a DIM. 
The combination is invalid for this case.
\end{enumerate}

Both combinations are invalid
\end{description}

According to the previous cases there are at most two extra solutions that should be compared to other candidate solutions. One of this solutions is given by 
taking $E'=E$, and the other is a DIM of minimum weight).

\vspace{4mm}


\item[(b)] It does not satisfy (a) and there are three arcs $A_v,A_w,A_z\in \mathcal{A}$ such that $A_v\cup A_w\cup A_z = C$. Let $E'\ne E$ a perfect edge dominating set such that is not a DIM. 


It is clear that $E'$ determines a 3-coloring of $V$ that verifies (P1) and (P2). Let $v,w,z$ the corresponding vertices to $A_v,A_w,A_z$.
The possible combinations of colors that satisfies (P1) and (P2) are:

\vspace{4mm} 
\begin{description}
\item[(black,black,black).] 
In this combination, every arc $A_u$ should have non empty intersection with one of $A_v,A_w,A_z$, then $u$ is not white. Thus $B=V$ and $E'=E$. Contradiction. 
This combination is not valid.

\vspace{4mm}
\item[(white,gray,gray),(gray,white,gray),(gray,gray,white).] 
These three combinations are symmetric so we analyze one of them.
Assume $u$ is a black vertex. Then $A_u \cap A_v = \emptyset$ by (P2). Also $A_u \cap A_w = \emptyset$ and $A_u \cap A_z = \emptyset$ since $w$ and $z$ are
gray adjacent vertices that should satisfy (P1). Thus $A_u$ can not be anywhere and $u$ is not a black vertex. Hence $B = \emptyset$. 
It is clear that $E'$ is a DIM and therefore this combination is invalid.

\end{description}

None of these combinations is valid, hence the best solution is among the candidates $E=E'$ and DIM of minimum weight.



\item[(c)] It satisfies neither (a) nor (b), in this case $\mathcal{M}$ is a Helly circular-arc model. 



In this case $G$ is Helly circular-arc graph and $\mathcal{M}$ a Helly circular-arc model of $G$. 
We compute  $min_{p \in C}|\mathcal{A}(p)|$ and $max_{p \in C}|\mathcal{A}(p)|$  in order to analyze the different subcases:

\vspace{4mm}
\begin{description}
\item \textbf{(I) $min_{p \in C}|\mathcal{A}(p)|=0$:} Then $\mathcal{M}$  in an interval model and $G$ an interval graph, hence a chordal graph.
In \cite{LuKT02} there is a linear time algorithm for MWPED for chordal graphs.

\vspace{4mm}

\item \textbf{(II) $max_{p \in C}|\mathcal{A}(p)|=q\geq 4$:} Thus there is a $K_q$ formed by vertices that corresponds to arcs that belongs to $\mathcal{A}(p)$, 
where $|\mathcal{A}(p)|=q$.
Let $v_1,\dots,v_q$ the vertices and $A_{v_1}=(s_{v_1},t_{v_1}),\dots,A_{v_q}=(s_{v_q},t_{v_q})$ the corresponding arcs. 
The vertices $v_1,\dots,v_q$ must be black vertices for any 3-coloring that corresponds to a perfect edge dominating set. We can replace each arc
$A_{v_i}, 1\leq i\leq q$ with two arcs $A_{v^-_i}=(s_{v_i},p-\epsilon)$ and $A_{v^+_i}=(p+\epsilon,t_{v_i})$.
Clearly, the resulting model $\mathcal{M}^*$ is an interval model because the point $p$ do not belongs to any arc.
The intersection graph $G^*=(V^*,E^*)$ from $\mathcal{M}^*$ is an interval graph where each vertex $v_i, 1 \leq i \leq q$ is replaced by two non-adjacent 
vertices $v^-_i$ y $v^+_i$ and the neighbors of $v_i$ in $G$ ends being either neighbors of $v^-_i$ or neighbors of $v^+_i$ in $G^*$ since 
$\mathcal{M}$ do not contains two arcs that covers the entire circle $C$.
In addition, the vertices $v^-_1,\dots,v^-_q$ induces a $K_q$ in $G^*$ and the vertices $v^+_1,\dots,v^+_q$ induces another $K_q$ in $G^*$.
We can change the weight function $\omega$ to $\omega^*$ for $G^*$ for each $vw \in E$, in the following way:

$$\omega^*(vw) = \left\{\begin{array}{cl}
\omega(vw) & v,w\in V\setminus\{v_1,\dots, v_q\} \\
\omega(v_iw) & w\in V\setminus\{v_1,\dots, v_q\} _\wedge {v=v^-_i}, 1\leq i\leq q  \\
\omega(v_iw) & w\in V\setminus\{v_1,\dots, v_q\} _\wedge {v=v^+_i}, 1\leq i\leq q  \\
\omega(v_iv_j) & {v=v^-_i} _\wedge {w=v^-_j}, 1\leq i,j \leq q  \\
0 & \mbox{otherwise}
\end{array}\right.$$

We will show that each perfect edge dominating set $E'$ from $G$ corresponds to one perfect edge dominating set $E''$ from $G^*$ with the same weight.
Thus, the problem is reduced to solve MWPED for interval graph $G^*$ with the linear time algorithm \cite{LuKT02}.

Given a perfect edge dominating set $E'$ from $G$, we show how to generate $E''$. 
Let the 3-coloring of $G$ corresponding to $E'$, where $v_1,\dots,v_k$ are black vertices. 
In $G^*$ those vertices are replaced by $v^-_1,\dots,v^-_k$ and $v^+_1,\dots,v^+_k$.
The new 3-coloring for $G^*$ consists on assign as the black vertices to $v^-_1,\dots,v^-_k$ and $v^+_1,\dots,v^+_k$, while the original vertices maintain
the same color.
This coloring must satisfy (P1) and (P2) in $G^*$. 
Otherwise, assume (P1) is not satisfied by a vertex $w$. Clearly, $w \in G$ and verifies (P1).
Since white and gray vertices as their adjacencies remains in both graphs, then $w \in G^*$ contains exactly one black vertex $v_i$, $1 \leq i \leq q$ in $G$,
and in $G^*$ it has two black neighbors $v^-_i$ and $v^+_i$. This is a contradiction since in this case $A_w \cup A_{v_i}$ covers the entire circle $C$ from 
$\mathcal{M}$. Thus (P1) is satisfied.
On the other hand, every white vertex $w \in G$ verifies (P2). The vertices from $N(w)$ are gray and $N(w)$ remains the same at $G^*$, hence it satisfies (P2).
Thus the 3-coloring of $G^*$ corresponds to a perfect edge dominating set $E''$ from $G^*$. The sets $E'$ and $E''$ have the same weight.

Given a perfect edge dominating set $E''$ from $G^*$, there is a corresponding perfect edge dominating set $E'$ from $G$ with the same weight.

Clearly, in the 3-coloring corresponding to $E''$, every vertex $v^-_1,\dots,v^-_q$ and $v^+_1,\dots,v^+_q$ are black vertices because are part of $K_q$ with
$q \geq 4$. In order to generate the 3-coloring for $G$, assign $v_1,\dots,v_q$ as black vertices and let the remaining vertices with the same color.
It can be verified in a similar way that this 3-coloring verifies (P1) and (P2) from $G$, hence it corresponds to a perfect edge dominating set $E'$ from $G$.
Moreover the map of $E'$ and $E''$ is the same as the previous one.

\vspace{4mm}
\item \textbf{(III) $max_{p \in C}|\mathcal{A}(p)|=3$:}
This implies $G$ contains a $K_3$ formed by the vertices corresponding to $\mathcal{A}(p)$ where $|\mathcal{A}(p)|=3$. 

Let $v_1,v_2,v_3$ and their corresponding arcs $A_{v_1}=(s_{v_1},t_{v_1}),A_{v_2}=(s_{v_2},t_{v_2}),A_{v_3}=(s_{v_3},t_{v_3})$.
For every 3-coloring of $G$ that satisfies (P1) and (P2), those three vertices must be: (i) three black vertices or (ii) exactly one of them white and the other
two gray. 
We add an arc $(p-2*\epsilon, p+2*\epsilon)$ to the model, and call it $\mathcal{M}^+=(C,\mathcal{A}^+)$ and $G^+$ as his corresponding graph. 

It is easy to check that the 3-colorings of $G$ that satisfies (P1), (P2) and \textbf{(i)} have a one-to-one correspondence with 3-colorings of $G^+$ that
satisfies (P1) and (P2) since $v_1,v_2,v_3$ and the new vertex $s$ corresponding to the additional arc are black vertices since those four vertices form a $K_4$
and $s$ is not adjacent to any other vertex from $G^+$. 
In order to preserve the weight between the perfect edge dominating sets among $G$ and $G^+$, the new edges $v_1s, v_2s, $ and $v_3s$ must have weight 0.
Thus, MWPED can be solved for $G^+$ where the model $\mathcal{M}^+$ do not contain a set of two or three arcs that covers the entire circle $C$, is not an interval model and $max_{p \in C}|\mathcal{A}^+(p)|=4$, which can be solved using the previous case in linear time. An additional consideration for 3-colorings
of $G$ that satisfies (P1), (P2) and not (i) must be added.

For the 3-colorings that satisfies (P1), (P2) but not (i), each arc $A_{v_j}, 1 \leq i \leq 3$ can be replaced by two arcs $A_{v^-_i}=(s_{v_i},p-\epsilon)$ and $A_{v^+_i}=(p+\epsilon,t_{v_i})$ just as in (II), but instead of $q \geq 4$ arcs, now is just for 3 arcs.
Again, the model $\mathcal{M}^*$ is an interval model and the intersection graph $G^*=(V^*,E^*)$ is an interval graph where each vertex $v_i, 1 \leq i \leq 3$
is replaced by two non-adjacent vertices $v^-_i$ and $v^+_i$.

We define 3 different weighted functions $\omega_j, 1\leq j\leq 3$. For each edge $vw\in E^*$:
$$\omega_j(vw) = \left\{\begin{array}{cl}
\omega(vw) & v,w\in V\setminus\{v_1,\dots, v_3\} \\
\omega(v_iw) & w\in V\setminus\{v_1,\dots, v_3\} _\wedge {v=v^-_i}, 1\leq i\leq 3  \\
\omega(v_iw) & w\in V\setminus\{v_1,\dots, v_3\} _\wedge {v=v^+_i}, 1\leq i\leq 3  \\
\omega(v_kv_l) & v={v^-_k} _\wedge w={v^-_l} _\wedge \{k,l\}=\{1,2,3\}\setminus \{j\}\\
0 & v={v^+_k} _\wedge w={v^+_l} _\wedge \{k,l\}=\{1,2,3\}\setminus \{j\}\\
\infty & \mbox{otherwise}
\end{array}\right.$$

It can be easily checked that any perfect edge dominating set $E'$ from $G$ has a corresponding 3-coloring that satisfies (ii).

A 3-coloring of $G^*$ can be obtained from the previous one, by keeping the same colors for common vertices among both graphs, and ${v^-_i},{v^+_i}$ get colors from $v_i$, $1\leq i\leq 3$.
The 3-coloring obtained corresponds to a perfect edge dominating set $E''$ from $G^*$. If the 3-coloring of $G$ contains the vertex $v_j$, $j \in \{1,2,3\}$ with color white then the weight of $E'$ using $\omega$ is the same weight that $E''$ using $\omega_j$. 
Thus, the MWPED problem can be solved three times for the interval graph $G^*$, each time with a different function $\omega_j$, $1 \leq j \leq 3$, applying
linear algorithm from \cite{ChangL94a}.
If the best of these solutions has bounded weight then the solution forms a perfect edge dominating set of minimum weight in $G$ that satisfies (ii).
Otherwise, there is no perfect edge dominating set of $G$ whose three-coloring satisfies (ii). 
The minimum perfect edge dominating set from $G$ is the minimum returned from the solutions given by \cite{ChangL94a} for $G^*$ , the solution from $G^+$, $E$ and
the MWEED for $G$. All of them can be obtained in linear time.

\vspace{4mm}

\item \textbf{(IV) $max_{p \in C}|\mathcal{A}(p)|=2$ y $min_{p \in C}|\mathcal{A}(p)|\geq 1$:} Since $max_{p \in C}|\mathcal{A}(p)|>min_{p \in C}|\mathcal{A}(p)|$, then $min_{p \in C}|\mathcal{A}(p)|=1$. 
In this case, we can find the induced cycle $C_k$ with $k \geq 4$ of $G$ such that the arcs corresponding to the vertices covers the circle $C$ in $\mathcal{M}$
and the rest of arcs from $\mathcal{A}$ are pairwise disjoint and are contained in exactly one of the arcs from $C_k$.
The vertices from $C_k$ are $v_1,v_2,\dots, v_k$, $v_iv_{i+1}\in E$, $1\leq i\leq k-1$, y $v_kv_1\in E$, 
and the corresponding arcs of $v_1,v_2,\dots, v_k$ are $A_1=(s_1,t_1),A_2=(s_2,t_2),\dots,(s_k,t_k)$ and the circular order according to the extreme is:
$s_1,t_k,s_2,t_1,s_3,t_2,\dots,s_k$.

We analyze the following cases:

\begin{enumerate}

\item $G$ is $C_j$, thus $k=n$. In this case $L(G)$ is exactly$G$. Instead of a weighted function over the arcs, the weighted function is over the vertices.
(Note that $C_n$ is a circular-arc graph)

\item There exists $v_i$ with leaves. Without lost of generality, we assume $i=2$ and $w_1, \dots, w_{d(v_2)-2}$ are their leaves, $d(v_2) \geq 3$ is the degree of $v_2$.
In this case, the vertex $v_2$ is father of $w_1$, in any other 3-coloring corresponding to a perfect edge dominating set of $G$, $v_2$ is not white since otherwise $w_1$ would be gray but could not have a non white neighbor to satisfy (P1). Then $v_2$ must be black or gray.

\begin{enumerate}

\item $v_2$ is black:
Then $N(v_2)$ can not be all white, so $w_1,\dots,w_{d(v_2)-2}$ can not be white. But at the same time this set can not be gray because are adjacent only to $v_2$, and has no white vertices to satisfy (P1). Then $w_1, \dots, w_{d(v_2)-2}$ are black vertices. To solve this subcase, we can alterate the model in the following way:
Add two identical arcs  $(t_1,s_3)$. In this new model, any leave $w_j$, $v_2$ and the two new corresponding vertices to the additional arcs form a $K_4$, thus in 
any 3-coloring, these vertices must be black. But these two new vertices are adjacent only to $v_2$ and his leaves. Hence the remaining vertices can have any coloring
independent of these  new vertices. Therefore the perfect edge dominating sets from $G$ in this subcase corresponds one-to-one to the modified model.
The modified model do not contain a set of 2 or 3 arcs that cover the circle $C$ and there is a point $p$ such that is contained in 4 different arcs.
Algorithm (II) can be used to solve it.
By assigning weight 0 to the added arcs then the solutions can be mapped and the weight will be the same.

\item $v_2$ is gray:
Then it contains exactly one non-white vertex which must be one of $v_1, v_3, w_1,\dots,w_{d(v_2)-2}$ , and the rest are white vertices.
Assume $w_i$ is the non-white vertex, then it is black because there is no other neighbor, thus this 3-coloring satisfies (P1) and (P2).
A new 3-coloring can be obtained, swapping color of $w_i$ with the color of another leave $w_j$. 
The new coloring will hold properties (P1) and (P2), hence it is convenient to choose the leave $w_j$ such that $\omega(v_2w_j)$ is minimum.
Therefore there are 3 candidates to be considered as the non-white neighbor of $v_2$. These candidates are $v_1,v_3$ and the best $w_j$.
For each case we may alter the model, adding an arc that intersects only to $A_2$ and the arc corresponding to the non-white vertex chosen.
The new added edges have weight $\infty$ in order to prevent that in the new model the added vertex is chosen as the non-white vertex of $v_2$.
This forces to choose the same non-white vertex from the original model. 
The new vertex has no other neighbors, hence it does not affect the coloring of the rest of the vertices. The resulting model has a point $p$ which is contained
by 3 different arcs, then the linear time algorithm (III) can be used to solve MWPED on this new model.
The algorithm must be run 3 times, once for each election of the non-white vertex for $v_2$.
\end{enumerate}

It applies at most once the algorithm from (II) and three time the algorithm from (III), hence the algorithm time complexity is $O(n+m)$.
\end{enumerate}
\end{description}


\end{description}

%
%
%
%
%
\bibliographystyle{acm} 



\end{document}